\documentclass[letterpaper, 10 pt, conference]{ieeeconf}

\IEEEoverridecommandlockouts                              % This command is only needed if 
\usepackage{cite}
\usepackage{amsmath,amssymb,amsfonts}
\usepackage{graphicx}
\usepackage{textcomp}

\usepackage{graphicx}
\usepackage{subfigure}
\usepackage{epsfig} % for postscript graphics files
\usepackage{cite}

\usepackage{amsmath} 
\usepackage{amssymb}
\usepackage{hyperref}
\usepackage{mathtools}
\usepackage{optidef}
\usepackage{subcaption}
\usepackage{booktabs}
\usepackage{siunitx}
\usepackage{acro}
\usepackage{placeins}

\DeclareAcronym{bla}{
  short = BLA,
  long  = \emph{best linear approximation}
}
\DeclareAcronym{nrmse}{
  short = NRMSE,
  long  = normalized root mean square error
}
\DeclareAcronym{mse}{
  short = MSE,
  long  = mean square error
}
\DeclareAcronym{rmse}{
  short = RMSE,
  long  = root mean square error
}
\DeclareAcronym{nllfr}{
  short = NL-LFR,
  long  = \emph{nonlinear linear fractional representation}
}
\DeclareAcronym{pnlss}{
  short = PNLSS,
  long  = polynomial nonlinear state-space
}
\DeclareAcronym{rms}{
  short = RMS,
  long  = root mean square
}
\DeclareAcronym{lti}{
  short = LTI,
  long  = linear time-invariant
}
\DeclareAcronym{grs}{
    short = GRS,
    long  = guided residual search
}
\DeclareAcronym{mad}{
    short = MAD,
    long  = mean absolute deviation
}

\let\oldproof\proof
\let\endoldproof\endproof
\let\proof\relax
\let\endproof\relax
\usepackage{amsthm}
\let\proof\oldproof
\let\endproof\endoldproof

\theoremstyle{definition}
\newtheorem{assumption}{Assumption}

\newtheorem{remark}{Remark}
\newtheorem{theorem}{Theorem}

\newlength{\Slabelwidth}
\newcounter{Sctr}
\renewcommand{\theSctr}{\textbf{S\arabic{Sctr}}}

\newenvironment{Slist}{%
  \begin{list}{\textbf{\theSctr}}{%
      \usecounter{Sctr}%
      \setlength{\labelwidth}{\Slabelwidth}%
      \setlength{\labelsep}{0.6em}%
      \setlength{\leftmargin}{\dimexpr\labelwidth+\labelsep\relax}%
      \setlength{\itemindent}{0pt}%
      \setlength{\listparindent}{0pt}%
      \setlength{\parsep}{0pt}%
      \setlength{\topsep}{0.2em}%
      \setlength{\itemsep}{0.2em}%
  }%
}{%
  \end{list}%
}

\usepackage{fancyhdr}

\begin{document}

\title{\LARGE \bf
A guided residual search for nonlinear state-space identification
}

\author{Merijn Floren$^{1,2}$ and Jan Swevers$^{1,2}$
\thanks{
This work is supported by Flanders Make's IRVA projects ASSISStaNT and CoMoDO. $^{1}$Department of Mechanical Engineering, KU Leuven, Belgium. $^{2}$Flanders Make@KU Leuven, Belgium. Corresponding author email: \texttt{merijn.floren@kuleuven.be}.
}
}

\maketitle
\thispagestyle{empty}

%%%%%%%%%%%%%%%%%%%%%%%%%%%%%%%%%%%%%%%%%%%%%%%%%%%%%%%%%%%%%%%%%%%%%%%%%%%%%%%%
\begin{abstract}
    Identifying the parameters of nonlinear state-space models from input-output data typically requires solving a highly non-convex optimization problem, which is prone to slow convergence and suboptimal local solutions.
    This work improves the reliability and efficiency of the estimation process by decomposing the overall optimization problem into a sequence of tractable subproblems.
    Starting from a linear baseline model, nonlinear residual dynamics are first estimated using a \ac{grs} and subsequently refined through multiple-shooting optimization.
    Experiments on two benchmarks show competitive performance with state-of-the-art black-box methods and improved convergence over naive initialization.
\end{abstract}

\begin{keywords}
Nonlinear systems identification, identification for control, optimization
\end{keywords}

\section{Introduction}
\label{sec:introduction}

State-space models offer a flexible framework for nonlinear system identification, as they closely reflect the dynamical structure of physical systems while allowing structural constraints to be imposed.
A particularly structured formulation is the \ac{nllfr}, which separates dominant linear dynamics from localized nonlinearities through a feedback interconnection between a \ac{lti} system and a static nonlinearity, as illustrated in Fig.~\ref{fig:model_structure}. This formulation captures a wide range of nonlinear phenomena while retaining exploitable structure for efficient parameter estimation.

The focus of this work is on estimating \ac{nllfr} \emph{simulation} models from experimental input-output data, a task that typically requires solving a computationally demanding, highly nonlinear, and non-convex parameter optimization problem, which is sensitive to initialization and prone to poor local minima.
These recursion-induced challenges intensify for longer horizons and more complex models, underscoring the need for robust and efficient optimization strategies.

One way to alleviate these long-horizon difficulties is to optimize over shorter simulation intervals. Doing so smooths the loss landscape~\cite{ribeiro2020smoothness} and allows for inherent parallelization, but requires reliable state estimates at the beginning of each interval.
The SUBNET approach~\cite{beintema2021nonlinear}, for example, infers these initial states using a deep subspace encoder trained jointly with the nonlinear model, 
whereas multiple-shooting formulations~\cite{ribeiro2020smoothness,retzler2022shooting} treat the states as decision variables, thereby introducing an additional initialization problem.

\begin{figure}[t]
    \centering
    \includegraphics[scale=1]{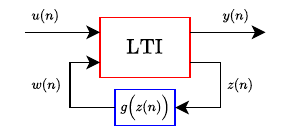}
    \caption{Schematic overview of the \acs{nllfr} structure.}
    \label{fig:model_structure}
    % \vspace*{-18pt}
\end{figure}

State estimation can also be addressed by first identifying a linear surrogate model~\cite{paduart2010identification,schoukens2020initialization,marconato2013improved,floren2025inference,floren2026latent,floren2022nonlinear}. In~\cite{marconato2013improved}, this linear approximation serves to estimate the nonlinear state trajectory through a least-squares trade-off between model consistency and data fit.
The inferred samples are then used to parametrize the nonlinear residual without recursion, yielding a fully initialized model for subsequent simulation-based refinement.
Guided residual search (\ac{grs}) methods~\cite{floren2025inference,floren2026latent,floren2022nonlinear,floren2024identification} follow the same sequential rationale, but exploit the \ac{nllfr} structure more explicitly by treating \(w\) in Fig.~\ref{fig:model_structure} as an auxiliary input to an extended linear model. 
However, both~\cite{marconato2013improved} and the time-domain \ac{grs} approaches~\cite{floren2026latent,floren2022nonlinear,floren2024identification} can induce a \emph{distribution shift} between the residual-parametrization stage and the simulation-based optimization stage, potentially hampering convergence of the latter.
So far, this shift has not been analyzed, demonstrated, or explicitly addressed.

This work characterizes the distribution shift both theoretically and experimentally, and proposes a mitigation strategy based on \ac{grs}-initialized multiple shooting, thereby avoiding the state-initialization burden of~\cite{ribeiro2020smoothness,retzler2022shooting}. In addition, we generalize the \ac{grs} rationale beyond existing formulations: unlike~\cite{floren2026latent,floren2022nonlinear,floren2024identification}, the proposed method requires no prior knowledge of the governing physical equations, and unlike~\cite{floren2025inference,floren2026latent,floren2022nonlinear}, it does not rely on periodic data.

The remainder of this paper is structured as follows. Section~\ref{sec:problem_statement} establishes the problem context, while Section~\ref{sec:method} presents the proposed identification approach. Section~\ref{sec:results} demonstrates the method's effectiveness on two experimental benchmark datasets, and Section~\ref{sec:conclusions} concludes the paper.

\textit{Notation:}
The sets of real, integer, and natural numbers are denoted by \( \mathbb{R} \), \( \mathbb{Z} \), and \( \mathbb{N} \), respectively.
For \( x \in \mathbb{R}^n \) and symmetric positive definite \( Q \in \mathbb{R}^{n \times n} \), define \(\|x\|_Q^2 = x^\top Q x\), and \(\|x\|^2 = x^\top x\), where \((\cdot)^\top\) denotes the transpose.
The identity matrix is denoted by \( I \), and \(0\) denotes the zero matrix, with dimensions clear from context.

% ----------------------------------------------------------------------------------
\section{Problem Statement}\label{sec:problem_statement}
The \ac{nllfr} structure in Fig.~\ref{fig:model_structure} is parametrized as:
\begin{subequations}
    \begin{align}
      x(n+1) &= A x(n) + B_u u(n) + B_w w(n),\label{eq:mod_x4}\\
      y_0(n) &= C_y x(n) + D_{yu} u(n) + D_{yw} w(n),\label{eq:mod_y4}\\
      z(n)   &= C_z x(n) + D_{zu} u(n),\label{eq:mod_z4}\\
      w(n)   &= f_{\text{NN}}\big(z(n)\big),\label{eq:mod_w4}
    \end{align}
    \label{eq:model4}%
\end{subequations}
where \(x(n) \in \mathbb{R}^{n_x}\) denotes the latent state at time index \(n \in \mathbb{Z}\), \(u(n)\in \mathbb{R}^{n_u}\) and \(y_0(n) \in \mathbb{R}^{n_y}\) denote the exact input and output, respectively, and \(z(n) \in \mathbb{R}^{n_z}\) and \(w(n) \in \mathbb{R}^{n_w}\) denote the respective input and output of the nonlinear function \(f_{\text{NN}}: \mathbb{R}^{n_z} \to \mathbb{R}^{n_w}\). The matrices \(A \in \mathbb{R}^{n_x \times n_x}\), \(B_u \in \mathbb{R}^{n_x \times n_u}\), \(B_w \in \mathbb{R}^{n_x \times n_w}\), \(C_y \in \mathbb{R}^{n_y \times n_x}\), \(C_z \in \mathbb{R}^{n_z \times n_x}\), \(D_{yu} \in \mathbb{R}^{n_y \times n_u}\), \(D_{yw} \in \mathbb{R}^{n_y \times n_w}\), and \(D_{zu} \in \mathbb{R}^{n_z \times n_u}\) define the \ac{lti} dynamics. The static nonlinearity \(f_{\text{NN}}\) is parametrized as a feedforward neural network with \(n_l\) hidden layers:
\begin{equation}
\begin{aligned}
z^{[l]}(n) &= \sigma\big(W^{[l]} z^{[l-1]}(n) + b^{[l]}\big),\quad l=1,\,\dots,\,n_l, \\
w(n) &= W^{[n_l+1]} z^{[n_l]}(n) + b^{[n_l+1]},
\end{aligned}\vspace*{-4pt}
\end{equation}
where \(\sigma(\cdot)\) denotes an elementwise nonlinear activation function, and \(z^{[0]}(n) = z(n)\).
The weight matrices \(W^{[l]}\) and bias vectors \(b^{[l]}\), of appropriate dimensions, are collectively denoted by the parameter set \(\theta_{\text{NN}} = \{W^{[l]},\, b^{[l]}\}_{l=1}^{n_l+1}\).

For parameter estimation, we consider the dataset:
\begin{equation}
\mathcal{D} = \big\{\big(u(n),\,y(n)\big)\big\}_{n=0}^{N-1},
\end{equation}
where both \(u(n)\) and \(y(n)\) have zero mean and unit variance.\vspace*{-4pt}
\begin{assumption}
The measured output is corrupted by additive, zero-mean, stationary noise \(v(n)\) with finite variance, i.e., \(y(n) \coloneqq y_0(n) + v(n)\). The noise may be colored and is uncorrelated with the input \(u(n)\).
\vspace*{-4pt}
\end{assumption}
\begin{assumption}
The input signal \(u(n)\) is exactly known and is persistently exciting of a sufficiently high order.
\vspace*{-4pt}
\end{assumption}
\begin{assumption}\label{ass:init_model}
An initial model
\(\mathcal{M}_0 = (\tilde{A}, \tilde{B}_u, \tilde{C}_y, \tilde{D}_{yu})\) 
is available. The model is stable and captures the dominant linear input-output behavior of the system.
\vspace*{-4pt}
\end{assumption}
The construction of \(\mathcal{M}_0\) is omitted for brevity; standard linear identification procedures can be used, e.g.,~\cite{schoukens2020initialization,floren2025inference}.

We propose a sequential identification strategy consisting of three stages: (i) performing a \ac{grs} to infer the latent variables \(w\) and \(x\) while simultaneously estimating \(B_w\) and \(D_{yw}\); (ii) {training the nonlinear parameters \(\theta_{\text{nl}}=(C_z,\,D_{zu},\,\theta_{\text{NN}})\) in a recursion-free manner} using the inferred latent variables; and (iii) jointly refining all model parameters through nonlinear optimization. Steps~(i) and~(ii) aim to accelerate convergence of step~(iii) by providing a well-informed initialization, but also introduce their own challenges, discussed next.

% \vspace*{-10pt}
\subsection{Mitigating distribution shift using multiple shooting}\label{subsec:distribution_shift}
Let \(\big\{\big(w_*(n),\,x_*(n),\,u(n)\big)\big\}_{n=0}^{N-1}\) denote the stage~(ii) training data, with \(w_*(n)\) and \(x_*(n)\) inferred by the \ac{grs} in stage~(i).
The estimate \(\hat{\theta}_{\text{nl}}\) is obtained by learning the mapping between~\eqref{eq:mod_z4} and~\eqref{eq:mod_w4} from these data in a recursion-free fashion.
Combining \(\hat{\theta}_{\text{nl}}\) with the baseline model \(\mathcal{M}_0\) and the stage~(i) estimates of \(B_w\) and \(D_{yw}\) yields the fully parametrized state-update mapping \(\hat{f}: \mathbb{R}^{n_x} \times \mathbb{R}^{n_u} \to \mathbb{R}^{n_x}\) corresponding to~\eqref{eq:mod_x4}, \eqref{eq:mod_z4}, and \eqref{eq:mod_w4}.
Now, when this mapping is deployed recursively as \(\hat{x}(n+1)=\hat{f}\big(\hat{x}(n),u(n)\big)\), small one-step discrepancies may accumulate over time, causing the simulated states \(\hat{x}(n)\) to drift away from the training states \(x_*(n)\), a phenomenon referred to as \emph{distribution shift}.  Building on the analysis in~\cite{venkatraman2015improving}, we formalize this effect as follows.\vspace*{-4pt}

\begin{theorem}
    Assume that (i) \(\hat{f}\) is uniformly Lipschitz continuous in \(x\) with constant \(L>1\), i.e., \(\|\hat{f}(x_1,\,u)-\hat{f}(x_2,\,u)\|\leq L\|x_1-x_2\|\) for all \(x_1,x_2\in\mathbb{R}^{n_x}\) and \(u\in\mathbb{R}^{n_u}\), and (ii) the one-step discrepancy satisfies \(\|\hat{f}\big(x_*(n),\,u(n)\big) - x_*(n+1)\|\leq \varepsilon\) for all \(n\in\{0,\dots,N-2\}\).
    Then, by recursively propagating the simulated states through \(\hat{f}\), with \(\hat{x}(0)=x_*(0)\), the deviation from the inferred trajectory at the final time sample satisfies
    \(
    \big\|\hat{x}(N-1)- x_*(N-1)\big\| \in \mathcal{O}(L^{N-1} \varepsilon).
    \)
    \label{th:distribution_shift}
\end{theorem}
\vspace*{-14pt}
\begin{proof}
    Define the deviation as \({e(n)} \coloneqq \hat{x}(n)-x_*(n)\). Then
    \(
      {e(n+1)}= \hat{f}\big(\hat{x}(n), \,u(n)\big) - x_*(n+1).
    \)
    Inserting \(\hat{f}\big(x_*(n),\,u(n)\big)\) twice and applying the triangle inequality yields
    \(
      \|{e(n+1)}\| \leq \|\hat{f}\big(\hat{x}(n), \,u(n)\big) - \hat{f}\big(x_*(n),\,u(n)\big)\| 
                    + \|\hat{f}\big(x_*(n),\,u(n)\big) - x_*(n+1)\|.
    \)
    By Lipschitz continuity and the one-step discrepancy bound, we obtain
    \(\|{e(n+1)}\|\leq L \|{e(n)}\| + \varepsilon\).
    With \({e(0)}=0\), iterating this recursion gives
    \(\|{e(N-1)}\| \leq \varepsilon \sum_{n=0}^{N-2} L^n\).
    For \(L>1\), the geometric sum equals \(\frac{L^{N-1}-1}{L-1}\), implying \(\|{e(N-1)}\|\in\mathcal{O}(L^{N-1}\varepsilon)\).
\end{proof}
{In other words, small one-step discrepancies can, in the worst case, be amplified geometrically during recursive simulation.
The exponential dependence on the horizon is worst-case tight under the stated assumptions\footnote{We assume \(L>1\) since \(\hat{f}\) is partly implemented using a neural network, which typically has a Lipschitz constant greater than one~\cite{gouk2021regularisation,khromov2023some}.}. Indeed, the scalar construction \(x_*(n+1)=Lx_*(n)\), with \(L >1\) and \(\hat f(x,u)=Lx+\varepsilon\), satisfies the assumptions of Theorem~\ref{th:distribution_shift} and yields \(e(n+1)=Le(n)+\varepsilon\), with \(e(0)=0\). Hence, \(e(N-1)=\varepsilon\sum_{n=0}^{N-2}L^n\in\Theta(L^{N-1}\varepsilon)\), which shows that the geometric accumulation can be attained exactly.}

{
Theorem~\ref{th:distribution_shift} highlights how the recursion-free training in stage~(ii) can impair subsequent simulation-based optimization by placing the optimizer far from a desirable local minimum, thereby offsetting the benefit of the \ac{grs}. 
At the same time, the error bound also suggests that this effect can be controlled via the horizon length.
Therefore, instead of optimizing the full trajectory from a single initial condition, we partition it into shorter non-overlapping intervals, so that the bound applies locally on each interval. The key requirement for this construction is the availability of an initial state estimate at the beginning of each interval, which is precisely provided by the inferred states \(x_*(n)\) from the \ac{grs} in stage~(i).}

{In the limiting case of unit-length intervals, local error accumulation is completely eliminated.
Without additional constraints, however, this would amount to one-step prediction-error minimization, which suffers from the same deployment issue described in Theorem~\ref{th:distribution_shift}.
To retain a simulation-error objective, the intermediate states are introduced as decision variables, and equality constraints are imposed to enforce continuity between consecutive intervals\footnote{This effectively eliminates the discrepancy term \(\varepsilon\) up to solver tolerances.}.
Doing so yields a smoother optimization problem with improved gradient information~\cite{ribeiro2020smoothness}, while allowing computations to be parallelized across intervals.
This formulation is known as \emph{multiple shooting} and is adopted in stage~(iii).}

% ----------------------------------------------------------------------------------
\section{Methodology}\label{sec:method}
\subsection{Bilevel guided residual search}\label{sec:bilevel_grs}
In this first stage, the model \(\mathcal{M}_0\) is held fixed, and the dual objective is to parametrize \(\theta_{wy} = (B_w,\,D_{yw})\) and infer the latent variables \(w\) and \(x\), used to train the neural network in Section~\ref{sec:neuralnet}. 
To this end, we propose an optimization scheme consisting of two levels: (i) an inner problem that estimates \(w\) and \(x\) conditioned on \(\theta_{wy}\), and (ii) an outer problem that updates \(\theta_{wy}\) based on the inferred signals. The scheme alternates between these levels until convergence.

\subsubsection{Inner problem}
Consider the fully parametrized linear submodel \(\mathcal{M}_1\) defined by~\eqref{eq:mod_x4} and \eqref{eq:mod_y4}, in which \(w\) is treated as an exogenous input.
The fixed matrices \(A\), \(B_u\), \(C_y\), and \(D_{yu}\) of \(\mathcal{M}_1\), in that order, are obtained from \(\mathcal{M}_0\) as
\begin{equation}
    \theta_{uy} = \big(T_x^{-1} \tilde{A} T_x,\, T_x^{-1} \tilde{B}_u,\, \tilde{C}_y T_x,\, \tilde{D}_{yu}\big),
    \label{eq:normalize}
\end{equation}
where \(T_x\) is diagonal and constructed from the empirical standard deviations of the simulated state trajectories of \(\mathcal{M}_0\).
The similarity transformation in~\eqref{eq:normalize} normalizes the state variables to have approximately unit variance, which is beneficial for all subsequent optimization stages.

The objective of this inner problem is to obtain a nonparametric estimate of 
\(w\) that minimizes the discrepancy between the simulated outputs of \(\mathcal{M}_1\) and the measured outputs in \(\mathcal{D}\).
For this purpose, we employ a sliding-window strategy, in which \(w\) is estimated by solving a local optimization problem over \(H+1\) samples, where \(H \ll N\) denotes the prediction horizon.
Only the solution corresponding to the first time instance is retained and used to shift the window forward. This procedure also yields an estimate of the latent state \(x\).\vspace*{-4pt}

\begin{assumption}\label{ass:outputcontrollability}
The submodel \(\mathcal{M}_1\) is output controllable with respect to its exogenous input \(w\).
\vspace*{-4pt}
\end{assumption}
\begin{remark}\label{rem:window_length}
{The windowed formulation is structurally similar to \emph{moving horizon estimation} and reduces sensitivity to (colored) noise in \(\mathcal{D}\) by estimating \(w\) from multiple samples rather than instantaneous or one-step-ahead relations.}
\vspace*{-4pt}
\end{remark}

In the following, the explicit dependence on \(\theta_{wy}\) is omitted for brevity.
For \(s \in \{u,w,y\}\), we first define the vectors
\begin{equation}
\mathcal{S}_s(n)
=
\begin{bmatrix}
s(n)^\top &
s(n+1)^\top &
\cdots &
s(n+H)^\top
\end{bmatrix}^\top,
\end{equation}
of length \(n_s (H+1)\).
Then, for \(n=0,\ldots,N-H-1\), we solve:
\begin{mini!}|s|[2]
    {\mathcal{S}_w(n)}{
        \frac{1}{2}\big\| \mathcal{S}_y(n) - \hat{\mathcal{Y}}(n) \big\|^2 
        + \frac{\lambda}{2} \big\| \mathcal{S}_w(n) \big\|_{Q_w}^2,\label{eq:stacked_cost}}
    {\label{eq:opti_problem}}{}
    \addConstraint{{\hat{\mathcal{Y}}}(n)}{=\mathcal{O}_x x_*(n) + \mathcal{T}_u \mathcal{S}_u(n) + \mathcal{T}_w \mathcal{S}_w(n),\label{eq:stacked_output}}
\end{mini!}
where \(\lambda \in \mathbb{R}_{>0}\) regularizes the solution through
\begin{equation}
Q_w
=
\begin{bmatrix}
B_w \\
D_{yw}
\end{bmatrix}^{\top}
\begin{bmatrix}
B_w \\
D_{yw}
\end{bmatrix}
+
\frac{\epsilon}{\lambda} I,
\end{equation}
with \(\epsilon \in \mathbb{R}_{>0}\) a small constant ensuring strict positive definiteness. 
This specific choice of \(Q_w\) penalizes the influence of \(\mathcal{S}_w(n)\) in terms of the effect on the system rather than its raw magnitude, thus making the regularization scale-invariant with respect to \(B_w\) and \(D_{yw}\), so that \(\lambda\) retains its intended meaning regardless of how these matrices are parametrized.

The stacked output predictions in~\eqref{eq:stacked_output} are expressed using the extended observability matrix
\begin{subequations}
    \begin{equation}
        \mathcal{O}_x=
            \begin{bmatrix}
            (C_y)^\top & (C_y A)^\top & \cdots & (C_y A^H)^\top
            \end{bmatrix}^{\top},
    \end{equation}
    of size \(n_y (H+1) \times n_x\), and the block Toeplitz matrices
    \begin{equation}
        \mathcal{T}_v=
            \begin{bmatrix}
                D_{yv} & 0 & \cdots & 0 \\
                C_y B_v & D_{yv} & \cdots & 0\\
                \vdots & \vdots & \ddots & \vdots\\
                C_y A^{H-1}B_v & C_y A^{H-2}B_v & \cdots & D_{yv}
            \end{bmatrix},
    \end{equation}
\end{subequations}
for \(v \in \{u,w\}\), with \(\mathcal{T}_v \in \mathbb{R}^{n_y (H+1) \times n_v (H+1)}\).
The optimization problem in~\eqref{eq:opti_problem} is convex, with closed-form solution:
\begin{equation}
    \mathcal{S}_{w_*}(n) = -\mathcal{G}^{-1}\mathcal{T}_w^\top\big(\mathcal{O}_x x_*(n) + \mathcal{T}_u \mathcal{S}_u(n) - \mathcal{S}_y(n)\big),
    \label{eq:opti_W}
\end{equation}
where \( \mathcal{G} = \mathcal{T}_w^\top \mathcal{T}_w + \lambda Q_w\). 
The first \(n_w\) elements of \(\mathcal{S}_{w_*}(n)\) are retained as \(w_*(n)\) and shift the window forward as
\begin{equation}
    x_*(n+1)= A x_*(n) + B_u u(n) + B_w w_*(n).
\end{equation}
In this recursion, \(x_*(0)\) is set to zero, thereby inducing a transient response that is unsuitable for parametric modeling of the nonlinear residual. Therefore, the first \(N_0 \in \mathbb{N}\) samples of each realization are later discarded, finally leading to
\begin{equation}
    \mathcal{D}_{*} = \big\{\big(w_*(n),\, x_*(n),\, {y}_*(n)\big)\big\}_{n=N_0}^{N-H-1},
    \label{eq:inferred_data2}
\end{equation}
where \({y}_*(n)\) is computed from~\eqref{eq:mod_y4} using the inferred \(x_*(n)\) and \(w_*(n)\), together with the current \(\theta_{wy}\) and the fixed \(\theta_{uy}\).

\subsubsection{Outer problem} \label{subsec:qop}
The dataset \(\mathcal{D}_{*}\) is obtained each time the inner problem is solved for a given \(\theta_{wy}\). The entries of \(\theta_{wy}\) are initialized from \(\mathcal{U}(-1,1)\) and updated as:
\begin{mini}|s|[1]
    { \theta_{wy}}
    {\frac{1}{N_{\text{tot}}} \sum_{n=N_0}^{N- H- 1}
       \big\| y(n) - {y}_*(n \mid \theta_{wy}) \big\|^2,\label{eq:outer_cost}}
    {}{}
\end{mini}
where \(N_{\text{tot}} = N-H-N_0\). 
\begin{remark}\label{rem:initsens}
    The outcome of the \ac{grs} is sensitive to the random initialization of \(\theta_{wy}\), as typical initializations render \(\mathcal{M}_1\) output controllable with respect to \(w\).
    This dependence introduces multiple effective minimizers of the outer loss~\eqref{eq:outer_cost}, not all yielding meaningful ({i.e., static}) relationships between \(x\) and \(w\). Especially for high-dimensional systems, multiple initializations should therefore be considered, with the neural network trained on each inferred dataset.
    \vspace*{-4pt}
\end{remark}
\begin{remark}
    {Stability of \(\mathcal{M}_1\) is inherited from Assumption~\ref{ass:init_model}.}
    \vspace*{-4pt}
\end{remark}

\subsection{Parametric learning of the nonlinear residual}\label{sec:neuralnet}
This second stage estimates the nonlinear residual parameters \(\theta_{\text{nl}}\) from \(\mathcal{D}_{*}\) obtained at convergence of the \ac{grs}.
The parameters are  obtained by solving
\begin{mini}|s|[1]
    { \theta_\text{nl}}
    {\frac{1}{N_{\text{tot}}} \sum_{n=N_0}^{N- H- 1}
       \big\| w_*(n) - \hat{w}_*(n \mid  \theta_\text{nl}) \big\|^2,\label{eq:nn_opti_cost}}
    {}{}
\end{mini}
where \(\hat{w}_*(n \mid  \theta_\text{nl})= f_{\text{NN}}\big(C_z x_*(n) + D_{zu} u(n)\big)\).
The matrices \(C_z\) and \(D_{zu}\) are initialized with elements drawn from \(\mathcal{U}(-1,1)\), while \(\theta_{\text{NN}}\) is initialized following~\cite{glorot2010understanding}.

\subsection{Final optimization using multiple shooting}\label{sec:final_opti}
In this final stage, we have (i) a fully initialized \ac{nllfr} model with parameters \(\theta =(\theta_{uy},\,\theta_{wy},\,\theta_{\text{nl}})\), and (ii) state estimates obtained from the \ac{grs}. This setup establishes a fully initialized multiple-shooting problem with both \(\theta\) and the shooting states as decision variables. Specifically, the starting indices of the shooting intervals are defined as
\begin{equation}
    \mathcal{I} \coloneqq \big\{N_0 + i d \mid i \in \{0, \dots, \lfloor (N_\text{tot}-1)/d \rfloor\}\big\},
\end{equation}
where \(d \geq 1\) denotes the interval length and \(\lfloor \cdot \rfloor\) is the floor function. 
The corresponding shooting states are defined as
\begin{equation}
    \mathcal{X} \coloneqq \big\{ x(n) \mid n \in \mathcal{I} \big\},
\end{equation}
with their initial values taken from \(\mathcal{D}_{*}\).
State propagation within each interval is enforced over the index set
\begin{equation}
    \mathcal{N}_i \coloneqq \{i, \, \dots, \, \min(i  +  d  -  1,\, N  -  H  -  1)\}, 
    \quad i \in \mathcal{I}.
\end{equation}
Finally, let \(\mathcal{I}^-\) denote the set of all but the last indices in \(\mathcal{I}\).
Then, the multiple-shooting problem becomes:
\begin{mini!}|s|[2]
    {\theta, \mathcal{X}}
    {\frac{1}{N_{\text{tot}}} \sum_{n=N_0}^{N- H- 1}
       \big\| y(n) - \hat{y}(n \mid \theta,\, \mathcal{X}) \big\|^2,}
    {\label{eq:opti_probledm_nonlin}}{\label{eq:nonlin_ids_cost}}
    \addConstraint{\hat{x}(n)}{= {x} (n),}
                  {n \in \mathcal{I} }
    \addConstraint{\hat{x}(n+1)}{= f \big(\hat{x}(n), \, u(n)\big),\,\,\,\,\,}
                  {n \in \mathcal{N}_i, \, i \in \mathcal{I} 
                  }
    \addConstraint{\hat{y}(n)}{ = h \big(\hat{x}(n), \, u(n)\big), }
                  {n \in  \mathcal{N}_i, \, i \in \mathcal{I} 
                  }
    \addConstraint{\hat{x}(n + d) }{ =  x(n + d), \label{eq:ms_cons}}
                  {n \in \mathcal{I}^-,}
\end{mini!}
with \(f(\cdot)\) and \(h(\cdot)\) denoting the unified state-update and output equations, respectively, according to~\eqref{eq:model4}. Continuity of the state trajectory across the intervals is ensured by~\eqref{eq:ms_cons}.

% ----------------------------------------------------------------------------------
\section{Experimental Results}\label{sec:results}
In this section, the proposed method is evaluated on two experimental benchmarks, followed by a discussion of computational aspects. Since both datasets contain periodic data, the respective initial linear models are parametrized using the \ac{bla}, following the procedure in~\cite{floren2025inference}; the subsequent steps do not exploit periodicity. All computations are performed in Python on a \SI{2.3}{\giga\hertz} Intel Core i7 with \SI{16}{\giga\byte} RAM. The \ac{grs} and neural network training are implemented in \texttt{JAX}~\cite{jax2018github}, using solvers from \texttt{Optax}~\cite{optax2020github} and \texttt{Optimistix}~\cite{optimistix2024}. The multiple-shooting optimization is implemented in \texttt{CasADi}~\cite{Andersson2019} using the IPOPT solver. {The code is available at \url{https://github.com/merijnfloren/guided-residual-search}.}

\subsection{Silverbox benchmark system}\label{sec:results_silverbox_mdof}
The Silverbox system~\cite{wigren2013three} represents an electronic implementation of a mass-spring-damper system with a cubic spring nonlinearity, with voltages representing force and displacement.
The input-output data exhibit an arrow-shaped structure consisting of two segments, sampled at approximately \SI{610}{\hertz}.
The first segment (“arrowhead”, used for testing) contains \num{40000} samples of white Gaussian noise with linearly varying amplitude;
the second segment consists of \num{86750} samples corresponding to ten random-phase multisine realizations, of which the final \num{21688} samples are reserved for testing.

\begin{figure}[!t]
    \centering
    \includegraphics[scale=1]{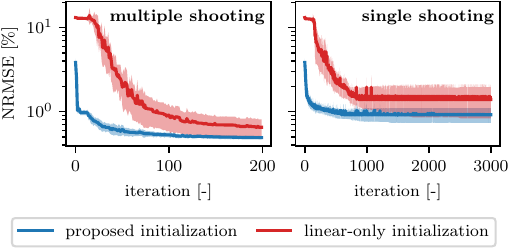}
    \caption{\ac{nllfr} simulation \acsp{nrmse} over the iterations of the final
    optimization stage {for Silverbox scenarios~\ref{scenario:S1}--\ref{scenario:S4}}. 
    }
    \label{fig:shooting_comparison}
    \vspace*{-18pt}
\end{figure}

\subsubsection{Algorithmic performance}
We first assess how \ac{grs}-based initialization improves the final optimization stage, and compare this effect for multiple-shooting and single-shooting strategies. We consider the following four scenarios.
\begin{Slist}
  \item \label{scenario:S1}\textit{Multiple shooting with \ac{grs}.}
  The procedure in Section~\ref{sec:method} is applied with \(n_x=2\) and \(n_w=n_z=1\). The \ac{grs} uses \(H=10\), \(N_0=100\), \(\lambda=1\), and \(\epsilon=10^{-8}\) for 25 Levenberg-Marquardt iterations. The single-layer neural network contains 15 ReLU neurons and is trained for 100 Adam iterations with learning rate \(1\times10^{-3}\). Multiple shooting uses \(d=1\) for 200 IPOPT iterations.
  \item \label{scenario:S2}\textit{Single shooting with \ac{grs}.}
  Identical initialization as in~\ref{scenario:S1}, followed by 3000 IPOPT iterations of single shooting\footnote{Multiple shooting typically converges in fewer iterations than single shooting due to reduced non-convexity, but requires more computational effort per iteration. We therefore compare IPOPT iteration counts corresponding to roughly the same runtime on our CPU (about \SI{20}{\minute}).}, implemented by setting \(d = N_{\text{tot}}\).
  \item \label{scenario:S3}\textit{Multiple shooting with linear-only initialization.} 
    The \ac{bla} model is normalized as in~\eqref{eq:normalize}, and \(\theta_{\text{nl}}\) is initialized according to Section~\ref{sec:neuralnet}, using the same architecture and initialization as in~\ref{scenario:S1}.
    Yet, the matrices \(B_w\) and \(D_{yw}\) are initialized from \(\mathcal{U}(-10^{-4}, 10^{-4})\), ensuring agreement with the \ac{bla} loss at the first iteration while avoiding zero gradients.
    The initial states for multiple shooting are obtained from simulated trajectories of the normalized \ac{bla}; all other settings follow~\ref{scenario:S1}.
   \item \label{scenario:S4}\textit{Single shooting with linear-only initialization.} Same initialization as in~\ref{scenario:S3}, followed by single-shooting optimization with the settings of~\ref{scenario:S2}.
\end{Slist}
We simulate each scenario 25 times, each with a different random initialization seed, and monitor the simulation performance of the \ac{nllfr} models across iterations of~\eqref{eq:opti_probledm_nonlin} using the \ac{nrmse}.
The results are visualized in Fig.~\ref{fig:shooting_comparison}, where the solid lines represent the median performance and the shaded areas indicate their respective \acp{mad}. Here, the \acp{nrmse} are computed over the full training set \(\mathcal{D}\) with the first 100 samples discarded to exclude transients.

As expected, the \ac{grs} provides a lower initial \ac{nrmse} at the start of the simulation-based optimization. The improved initialization also promotes consistent convergence to lower final \ac{nrmse} values, regardless of the shooting strategy. Moreover, the steep initial descent observed for scenarios~\ref{scenario:S1} and~\ref{scenario:S2}, in contrast to the limited progress in scenarios~\ref{scenario:S3} and~\ref{scenario:S4}, indicates that the former initializations place the parameters in a favorable region of attraction with informative gradients.
Finally, it is worth noting that the initialization overhead in~\ref{scenario:S1} and~\ref{scenario:S2} remains modest: the \ac{grs} and neural network training together required only \SI{21}{\second} on average.

Figure~\ref{fig:shooting_comparison} further demonstrates that multiple shooting outperforms single shooting, consistently converging to lower final \ac{nrmse} values regardless of the initialization approach.
This behavior likely results from the smoothing effect of multiple shooting on the loss landscape and its gradients~\cite{ribeiro2020smoothness}.

\begin{figure}[!t]
    \centering
    \includegraphics[scale=1]{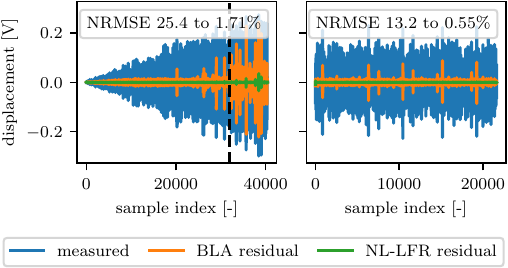}
    \caption{Best-performing \ac{nllfr} model from~\ref{scenario:S1} on the Silverbox arrowhead (left) and multisine (right) test data, compared to the \ac{bla} model. 
    }
    \label{fig:silverbox_test}
\end{figure}
\begin{table}
    \centering
    \caption{Simulation RMSEs on the Silverbox test data (mean \(\pm\) standard deviation in \SI{}{\milli\volt}) for the \ac{nllfr} models of~\ref{scenario:S1}, compared with state-of-the-art black-box methods.}
    {\setlength{\tabcolsep}{5pt} % default is 6pt
    {\fontsize{8}{8}\selectfont
    \begin{tabular}{l|cccc}
        & proposed & SUBNET~\cite{beintema2021nonlinear} & PNLSS~\cite{paduart2010identification}\\
    \midrule
    multisine          &  0.33 \(\pm\) 0.02 & 0.36  & -  \\
    arrowhead (full)       &  1.13 \(\pm\) 0.15 & 1.40  & 0.26 \\
    arrowhead (no extrap.) &  0.31 \(\pm\) 0.03 & 0.32  & - 
    \end{tabular}
    }}
    \label{tab:silverbox_results_grs}
    \vspace*{-10pt}
\end{table}
% \FloatBarrier

\subsubsection{Test data performance}
The best-performing model from scenario~\ref{scenario:S1} is evaluated on the test data, as illustrated in Fig.~\ref{fig:silverbox_test}. The resulting model accurately reproduces the system dynamics for both excitation signals and clearly outperforms the \ac{bla}.
For the arrowhead signal, most of the error arises from extrapolation beyond the training domain (black dashed line), which is expected given the black-box nature of neural networks and their limited extrapolation capability.

Table~\ref{tab:silverbox_results_grs} summarizes the \ac{rmse} test values of the models obtained from~\ref{scenario:S1}, alongside benchmark results from state-of-the-art black-box approaches: SUBNET~\cite{beintema2021nonlinear} and the \acs{pnlss} method~\cite{paduart2010identification}. The proposed method performs competitively, with SUBNET exhibiting similar extrapolation errors on the arrowhead signal. In contrast, \ac{pnlss} achieves superior extrapolation due to its degree-three polynomial basis, consistent with the cubic system nonlinearity.
In terms of computation, SUBNET reports a training time exceeding one day, while the proposed method requires only \SI{20}{\minute}; \acs{pnlss} does not report training times.

\subsection{F-16 ground vibration test}
\begin{figure}[!b]
    \vspace*{-8pt}
    \centering
    \includegraphics[scale=1]{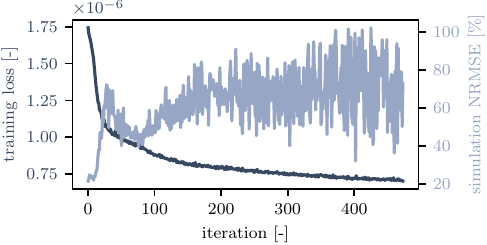}
    \caption{Distribution shift over the F-16 neural network training iterations. The
    static neural network loss~\eqref{eq:nn_opti_cost} steadily decreases, but the simulation error of the resulting \ac{nllfr} model quickly deteriorates and exhibits high volatility.}
    \label{fig:f16_dist_shift}
\end{figure}
We next consider a more complex F-16 ground vibration dataset~\cite{noel2017f}, comprising one input and three outputs. Excitation is applied using an electrodynamic shaker mounted beneath the right wing, with accelerometers measuring the response near the wing-tip payload interface and the excitation point. 
The estimation data consist of eight steady-state periods of a single random-phase multisine excitation spanning \SIrange{2}{15}{\hertz}, each containing \(8192\) samples and a \ac{rms} excitation amplitude of \SI{73.6}{\newton}. Prior to identification, the periods are averaged, and the data are low-pass filtered at \SI{15}{\hertz} to reduce the influence of measurement noise.
Based on empirical tuning, the model orders are selected as \(n_x=15\), \(n_w=2\), and \(n_z=6\). The \ac{grs} is performed with \(H=10\), \(N_0=500\), \(\lambda=100\), and \(\epsilon=10^{-8}\) for 50 Levenberg-Marquardt iterations. 
The neural network has one hidden layer with 32 \(\tanh\) units and is trained for up to 1000 Adam iterations with learning rate \(4\times10^{-3}\).

{Following Remark~\ref{rem:initsens}, we repeat the \ac{grs} and neural network training for 50 random initialization seeds. The combined runtime was measured as \(\SI{195}{\second} \pm \SI{1.71}{\second}\) (median \(\pm\) \ac{mad}), with approximately \SI{30}{\second} spent on the \ac{grs}, thereby indicating that the proposed initialization strategy remains lightweight enough to make a multi-start implementation practical.
}

To experimentally examine the distribution shift described in Theorem~\ref{th:distribution_shift}, we consider a typical neural network training run and track the deployment performance of the corresponding \ac{nllfr} model throughout the training iterations\footnote{
    Deployment performance is quantified by the simulation \ac{nrmse}, averaged over the three outputs, obtained if the full \ac{nllfr} model were to be simulated after each training iteration. 
}. The results, shown in Fig.~\ref{fig:f16_dist_shift}, indicate that although the training loss~\eqref{eq:nn_opti_cost} decreases monotonically, the deployment performance rapidly deteriorates and exhibits pronounced volatility. 
One plausible explanation, consistent with Theorem~\ref{th:distribution_shift}, is that while neural network training reduces the one-step-ahead error \(\varepsilon\), it may simultaneously increase the effective Lipschitz constant \(L\)~\cite{khromov2023some}.
The combined effect of these opposing trends may then substantially degrade simulation accuracy.

We initialize the final optimization stage with the \ac{nllfr} model obtained at the last iteration in Fig.~\ref{fig:f16_dist_shift} and compare the resulting multiple-shooting optimization over 1000 IPOPT iterations with \(d=200\) samples\footnote{Chosen as a compromise between memory usage and algorithm performance;  $d=1$ would offer best convergence but exceeds available memory.} against a linear-only initialization\footnote{
This essentially follows the approach of~\cite{schoukens2020initialization}, replacing single shooting by multiple shooting in the final optimization stage.
}. 
The resulting simulation \acp{nrmse}, averaged over the three outputs, are shown in Fig.~\ref{fig:benchmark_resultsf16_final_opti}. 
Despite starting from a lower simulation error, the linear-only initialization fails to escape a poor local minimum, whereas the proposed initialization improves rapidly and converges to a substantially better solution, indicating a more favorable initial parameter region.
{Each iteration in Fig.~\ref{fig:benchmark_resultsf16_final_opti} required approximately \SI{4.2}{\second}, resulting in a total runtime of about \SI{70}{\minute}.}

We finally assess the generalization performance of the optimized \ac{nllfr} model on an unseen test dataset consisting of eight steady-state periods of the same random-phase multisine realization, applied at a higher amplitude of \SI{85.7}{\newton} \ac{rms}. As with the training data, the periods are averaged and low-pass filtered at \SI{15}{\hertz}. 
Table~\ref{tab:f16_results_grs} reports the per-output simulation \acp{nrmse} compared to the \ac{bla}\footnote{Direct comparison with prior literature is limited by the sparse use of this benchmark and by differences in data processing across studies.}. 
The \ac{nllfr} model achieves a clear improvement, although the increased excitation amplitude, system complexity, and limited training data result in some residual simulation error. 
% \vspace*{-8pt}

\begin{figure}
    \centering
    \includegraphics[scale=1]{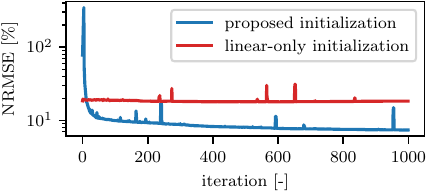}
    \caption{F-16 simulation \acp{nrmse} during the final optimization stage. Despite a higher initial error due to distribution shift (see Fig.~\ref{fig:f16_dist_shift}), the proposed method converges rapidly.
    }
    \label{fig:benchmark_resultsf16_final_opti}
\end{figure}

\begin{table}
    \centering
    \caption{Per-output simulation \acp{nrmse} of the \ac{bla} and \ac{nllfr} models evaluated on the F-16 test data.}
    {\setlength{\tabcolsep}{5pt} % default is 6pt
    {\fontsize{8}{8}\selectfont
    \begin{tabular}{l|cccc}
        & output 1 & output 2 & output 3 \\
    \midrule
    BLA        &  \SI{11.59}{\percent} & \SI{19.40}{\percent}  & \SI{20.59}{\percent}  \\
    NL-LFR      &  \SI{9.20}{\percent} & \SI{12.55}{\percent}  & \SI{13.27}{\percent} 
    \end{tabular}
    }}
    \label{tab:f16_results_grs}
    \vspace*{-10pt}
\end{table}

\subsection{Computational aspects}
{Both benchmark examples demonstrate the low computational cost of the proposed initialization stage. The dominant computational expense instead arises from the constrained multiple-shooting problem induced by the additional shooting-state decision variables\footnote{The present implementation did not exploit multiple-shooting parallelism.}. For the more complex F-16 example, this cost became a practical limitation: a compromise interval length of \(d=200\) samples was selected to keep the problem tractable. Moreover, this choice required empirical tuning, and reliable convergence was not obtained for all nearby hyperparameter settings or all random seeds.}

{
This computational bottleneck is particularly restrictive for high-dimensional systems with large effective Lipschitz constants, such as the F-16 benchmark, for which short shooting intervals are expected to improve robustness to hyperparameter choices and random-seed effects. Realizing this setting in practice will therefore require a more scalable multiple-shooting implementation that exploits the natural parallelism across intervals and improves memory management.}

% ----------------------------------------------------------------------------------
\section{Conclusions}\label{sec:conclusions}
{
This work introduced a computationally efficient \ac{grs}-based initialization scheme for simulation-based optimization of \ac{nllfr} state-space models. The adverse effect of distribution shift inherent to the proposed initialization was analyzed theoretically, illustrated experimentally, and mitigated through a multiple-shooting strategy. Experiments on two benchmarks showed improved convergence over naive initialization and competitive performance with state-of-the-art black-box methods. Future work will focus on more scalable multiple-shooting implementations.}

\bibliographystyle{IEEEtran}
\bibliography{IEEEabrv, references}

\end{document}